\documentclass[10pt, a4paper]{article}%
\usepackage{amsmath,amssymb,eucal}
\usepackage{amsfonts}
\usepackage{mathrsfs}
\usepackage{slashed}
\usepackage{graphicx}
\usepackage{ifpdf}
\usepackage{url}
\numberwithin{equation}{section} 
\def\beq{\begin{equation}}
\def\eeq{\end{equation}}

\def\bC{ {{\mathbb{C}}}}
\def\bR{ {{\mathbb{R}}}}

\def\Tr{ {{\rm{Tr}}} }

\newcommand{\pd}{{\rm DP}}
\newcommand{\sgn}{{\rm sgn}}
\newcommand{\gap}{{\rm gap}}
\newcommand{\pk}[1]{p_{\kappa}}

\newcommand{\ol}{\overline{l}}
\newcommand{\ul}{\underline{l}}
\newcommand{\oL}{\overline{L}}
\newcommand{\uL}{\underline{L}}
\newcommand{\on}{\overline{n}}
\newcommand{\un}{\underline{n}}

\newcommand{\bk}{\breve{\kappa}}

\newtheorem{defn}{{\bf Definition}}[section]
\newtheorem{thm}[defn]{{\bf Theorem}}
\newtheorem{cor}[defn]{{\bf Corollary}}
\newtheorem{lem}[defn]{{\bf Lemma}}

\newtheorem{rem}[defn]{{\bf Remark}}

\newtheorem{notation}[defn]{Notation}
\newenvironment{proof}[1][Proof]{\textbf{#1.} }{\hfill \rule{0.5em}{0.5em}}
\begin{document}

\title{Area Operator in Loop Quantum Gravity}
\author{Adrian P. C. Lim \\
Email: ppcube@gmail.com
}

\date{}

\maketitle

\begin{abstract}
A hyperlink is a finite set of non-intersecting simple closed curves in $\mathbb{R} \times \mathbb{R}^3$. Let $S$ be an orientable surface in $\mathbb{R}^3$. The dynamical variables in General Relativity are the vierbein $e$ and a $\mathfrak{su}(2)\times\mathfrak{su}(2)$-valued connection $\omega$. Together with Minkowski metric, $e$ will define a metric $g$ on the manifold. Denote $A_S(e)$ as the area of $S$, for a given choice of $e$.

The Einstein-Hilbert action $S(e,\omega)$ is defined on $e$ and $\omega$. We will quantize the area of the surface $S$ by integrating $A_S(e)$ against a holonomy operator of a hyperlink $L$, disjoint from $S$, and the exponential of the Einstein-Hilbert action, over the space of vierbeins $e$ and $\mathfrak{su}(2)\times\mathfrak{su}(2)$-valued connections $\omega$. Using our earlier work done on Chern-Simons path integrals in $\mathbb{R}^3$, we will write this infinite dimensional path integral as the limit of a sequence of Chern-Simons integrals. Our main result shows that the area operator can be computed from a link-surface diagram between $L$ and $S$. By assigning an irreducible representation of $\mathfrak{su}(2)\times\mathfrak{su}(2)$ to each component of $L$, the area operator gives the total net momentum impact on the surface $S$.
\end{abstract}

\hspace{.35cm}{\small {\bf MSC} 2010: } 83C45, 81S40, 81T45, 57R56 \\
\indent \hspace{.35cm}{\small {\bf Keywords}: Area, Path integral, Einstein-Hilbert, Loop representation, Quantum gravity}




\section{Quantization of Area}

Using canonical quantization of Ashtekar variables, the authors in \cite{rovelli1995discreteness} were able to quantize the area of a surface $S$ in $\bR^3$. Their procedure was to promote the variables into operators, which act on a quantum state, defined by a spin network.

The idea of a spin network was first introduced by Penrose, in an attempt to construct a quantum mechanical description of the geometry of space. A spin network in $\bR^3$ is a graph, each vertex has valency 3, with a positive integer assigned to every edge in the graph, satisfying certain conditions at each vertex.

What they showed was that a spin network $T$ is an eigenstate of the area operator and the eigenvalues of the area operator, is proportional to \beq \sum_l \sqrt{j_l(j_l+1)}, \nonumber \eeq whereby the sum is over all edges $l$ from the spin network $T$ crossing the surface $S$. Note that $j_l$ is a half-integer or integer greater than 0 and it represents the irreducible representation of $\mathfrak{su}(2)$. Thus, $\sqrt{j_l(j_l+1)}$ is the total angular momentum associated with each edge $l$ in consideration.

Unfortunately, there are some mathematical difficulties yet to be clearly explained, for example defining the domain of the operator, the topology for which the limit is to be taken etc. In this article, we are going to derive their result using a path integral approach. Instead of using spin network in $\bR^3$, we will use loops in $\bR^4$. And we will also show that both translational momentum and angular momentum appear as eigenvalues of the area operator as seen later in this article.


Consider $\bR^4 \equiv \bR \times \bR^3$, whereby $\bR$ will be referred to as the time-axis and $\bR^3$ is the spatial 3-dimensional Euclidean space. In future, when we write $\bR^3$, we refer to the spatial subspace in $\bR^4$. Let $\pi_0: \bR^4 \rightarrow \bR^3$ denote this projection. Fix the standard coordinates on $\bR^4\equiv \bR \times \bR^3 $, with time coordinate $x_0$ and spatial coordinates $(x_1, x_2, x_3)$.

\begin{notation}\label{n.s.1}(Subspaces in $\bR^4$)\\
Let $\{e_i\}_{i=1}^3$ be the standard basis in $\bR^3$. And $\Sigma_i$ is the plane in $\bR^3$, containing the origin, whose normal is given by $e_i$. So, $\Sigma_1$ is the $x_2-x_3$ plane, $\Sigma_2$ is the $x_3-x_1$ plane and finally $\Sigma_3$ is the $x_1-x_2$ plane.

Note that $\bR \times \Sigma_i \cong \bR^3$ is a 3-dimensional subspace in $\bR^4$. Here, we replace one of the axis in the spatial 3-dimensional Euclidean space with the time-axis. Let $\pi_i: \bR^4 \rightarrow \bR \times \Sigma_i$ denote this projection.
\end{notation}

For a finite set of non-intersecting simple closed curves in $\bR^3$ or in $\bR \times \Sigma_i$, we will refer to it as a link. If it has only one component, then this link will be referred to as a knot. A simple closed curve in $\bR^4$ will be referred to as a loop. A finite set of non-intersecting loops in $\bR^4$ will be referred to as a hyperlink in this article. We say a link or hyperlink is oriented if we assign an orientation to its components.

Let $L$ be a hyperlink. We say $L$ is a time-like hyperlink, if given any 2 distinct points $p\equiv (x_0, x_1, x_2, x_3), q\equiv (y_0, y_1, y_2, y_3) \in L$, $p \neq q$, we have
\begin{itemize}
  \item $\sum_{i=1}^3(x_i - y_i)^2 > 0$;
  \item if there exists $i, j$, $i \neq j$ such that $x_i = y_i$ and $x_j = y_j$, then $x_0 - y_0 \neq 0$.
\end{itemize}

\begin{rem}
In other words, a hyperlink $L$ is said to be time-like if $\pi_a(L)$, $a=0, 1, 2, 3$, are all links inside their respective 3-dimensional subspace $\pi_a(\bR^4)\subset \bR^4$.
\end{rem}

Throughout this article, all our hyperlinks in consideration will be time-like.
We adopt Einstein's summation convention, i.e. we sum over repeated superscripts and subscripts. Indices such as $a,b,c, d$ and greek indices such as $\mu,\gamma, \alpha, \beta$ will take values from 0 to 3; indices labeled $i, j, k$, $\bar{i}, \bar{j}, \bar{k}$ will only take values from 1 to 3. 

\section{Representation of $\mathfrak{su}(2) \times \mathfrak{su}(2)$}\label{s.su2}

Let $\mathfrak{su}(2)$ be the Lie Algebra of $SU(2)$. Here, we summarize the  representation of $\mathfrak{su}(2) \times \mathfrak{su}(2)$ as described in \cite{EH-Lim02}. We need to first choose a basis for $\mathfrak{su}(2)$ and we shall make the following choice,
\beq \breve{e}_1 :=
\frac{1}{2}\left(
  \begin{array}{cc}
    0 &\ 1 \\
    -1 &\ 0 \\
  \end{array}
\right),\ \ \breve{e}_2 :=
\frac{1}{2}\left(
  \begin{array}{cc}
    0 &\ i \\
    i &\ 0 \\
  \end{array}
\right),\ \ \breve{e}_3 :=
\frac{1}{2}\left(
  \begin{array}{cc}
    i &\ 0 \\
    0 &\ -i \\
  \end{array}
\right). \nonumber \eeq Using the above basis, define the following elements in $\mathfrak{su}(2) \times \mathfrak{su}(2)$,
\begin{align*}
\hat{E}^{01} = (\breve{e}_1, 0),\ \ \hat{E}^{02} = (\breve{e}_2, 0), \ \ \hat{E}^{03} = (\breve{e}_3, 0), \\
\hat{E}^{23} = (0, \breve{e}_1),\ \ \hat{E}^{31} = ( 0, \breve{e}_2), \ \ \hat{E}^{12} = (0, \breve{e}_3),
\end{align*}
and write
\beq \hat{E}^{\tau(1)} = \hat{E}^{23}, \ \ \hat{E}^{\tau(2)} = \hat{E}^{31}, \ \ \hat{E}^{\tau(3)} = \hat{E}^{12}. \nonumber \eeq

\begin{rem}
We will also write $\hat{E}^{\alpha\beta} = -\hat{E}^{\beta\alpha}$.
\end{rem}

Using the above basis, define \beq \mathcal{E}^\pm := \sum_{i=1}^3\breve{e}_i \in \mathfrak{su}(2), \label{e.sux.1} \eeq and a $4 \times 4$ complex matrix
\beq \mathcal{E} =
\left(
  \begin{array}{cc}
    -\mathcal{E}^+ &\ 0 \\
    0 &\ \mathcal{E}^- \\
  \end{array}
\right). \label{e.sux.2} \eeq

Let $\rho^\pm: \mathfrak{su}(2) \rightarrow {\rm End}(V^\pm)$ be an irreducible finite dimensional representation, indexed by half-integer and integer values $j_{\rho^\pm} \geq 0$. The representation $\rho: \mathfrak{su}(2) \times \mathfrak{su}(2) \rightarrow {\rm End}(V^+) \times {\rm End}(V^-)$ will be given by $\rho = (\rho^+, \rho^-)$, with \beq \rho: \alpha_i\hat{E}^{0i} + \beta_j \hat{E}^{\tau(j)} \mapsto \left(\sum_{i=1}^3\alpha_i \rho^+(\breve{e}_i) , \sum_{j=1}^3\beta_j \rho^-(\breve{e}_j) \right). \nonumber \eeq By abuse of notation, we will now write $\rho^+ \equiv (\rho^+, 0)$ and $\rho^- \equiv (0, \rho^-)$ in future and thus $\rho^+(\hat{E}^{0i}) \equiv \rho^+(\breve{e}_i)$, $\rho^-(\hat{E}^{\tau(j)}) \equiv \rho^-(\breve{e}_j)$.

Note that the dimension of $V^\pm$ is given by $2j_{\rho^\pm} + 1$.  Then it is known that the Casimir operator is
\begin{align*}
\sum_{i=1}^3 \rho^+(\hat{E}^{0i})\rho^+(\hat{E}^{0i})  = -\xi_{\rho^+} I_{\rho^+},\\
\sum_{i=1}^3 \rho^-(\hat{E}^{\tau(i)})\rho^-(\hat{E}^{\tau(i)}) = -\xi_{\rho^-} I_{\rho^-},
\end{align*}
$I_{\rho^\pm}$ is the $2j_{\rho^\pm} + 1$ identity operator for $V^\pm$ and $\xi_{\rho^\pm} := j_{\rho^\pm}(j_{\rho^\pm}+1)$.

Now, we will write for $a \in \bC$, the trace \beq \Tr_\rho \exp[a\mathcal{E}] \equiv \Tr_{\rho^+}\exp[a\rho^+(\mathcal{E}^+)] + \Tr_{\rho^-}\exp[a\rho^-(\mathcal{E}^-)]. \nonumber \eeq Without loss of generality, we assume that $\rho^\pm(\hat{E})$ is skew-Hermitian for any $\hat{E} \in \mathfrak{su}(2)$, so $\rho^+(i\mathcal{E}^+)$ has real eigenvalues. Thus, $\Tr\ \rho^{\pm}(e^{i\mathcal{E}^{\pm}}) \geq 1$ for any irreducible representation.

\begin{rem}
By choosing the group $SU(2) \times SU(2)$, we actually define a spin structure on $\bR^4$.
\end{rem}

\section{Area path integral}



Let $\overline{\mathcal{S}}_\kappa(\bR^4) \subset L^2(\bR^4)$ be a Schwartz space, as defined in \cite{EH-Lim02}. Using the standard coordinates on $\bR^4$, let $\Lambda^1(\bR^3)$ denote the subspace in $\Lambda^1(\bR^4)$ spanned by $\{dx_1, dx_2, dx_3\}$. Define
\begin{align*}
L_\omega :=& \overline{\mathcal{S}}_\kappa(\bR^4) \otimes \Lambda^1(\bR^3)\otimes \mathfrak{su}(2) \times \mathfrak{su}(2), \\
L_e :=& \overline{\mathcal{S}}_\kappa(\bR^4) \otimes \Lambda^1(\bR^3)\otimes V,
\end{align*}
whereby $\bR^4 \times V \rightarrow \bR^4$ is a trivial 4-dimensional vector bundle, with structure group $SO(3,1)$. This implies that $V$ is endowed with a Minkowski metric, $\eta^{ab}$, of signature $(-, +, +, +)$. Let $\{E^\gamma\}_{\gamma=0}^3$ be a basis for $V$.

Given $\omega \in L_\omega$ and $e \in L_e$, we will write
\begin{align*}
\omega =& A^i_{\alpha\beta} \otimes dx_i\otimes \hat{E}^{\alpha\beta} \in \overline{\mathcal{S}}_\kappa(\bR^4) \otimes \Lambda^1(\bR^3)\otimes \mathfrak{su}(2) \times \mathfrak{su}(2), \\
e =& B^i_\gamma \otimes dx_i\otimes E^\gamma \in \overline{\mathcal{S}}_\kappa(\bR^4) \otimes \Lambda^1(\bR^3) \otimes V.
\end{align*}
There is an implied sum over repeated indices.

\begin{rem}
Note that $A^{i}_{\alpha\beta} = -A^{i}_{\beta\alpha}\in \overline{\mathcal{S}}_\kappa(\bR^4)$.
\end{rem}

We define the Einstein-Hilbert action as
\begin{align*}
S_{EH}(e, \omega) :=
\frac{1}{8}&\int_{\bR^4}\epsilon^{abcd}B^1_\gamma B^2_\mu[E^{\gamma \mu}]_{ab} \cdot \partial_0 A^3_{\alpha\beta}[E^{\alpha\beta}]_{cd} dx_1\wedge dx_2 \wedge dx_0 \wedge dx_3\\
+& \frac{1}{8}\int_{\bR^4}\epsilon^{abcd}B^2_\gamma B^3_\mu[E^{\gamma \mu}]_{ab} \cdot \partial_0 A^1_{\alpha\beta}[E^{\alpha\beta}]_{cd} dx_2\wedge dx_3 \wedge dx_0 \wedge dx_1\\
+&\frac{1}{8}\int_{\bR^4}\epsilon^{abcd}B^3_\gamma B^1_\mu[E^{\gamma \mu}]_{ab} \cdot \partial_0 A^2_{\alpha\beta}[E^{\alpha\beta}]_{cd} dx_3\wedge dx_1 \wedge dx_0 \wedge dx_2.
\end{align*}
We sum over repeated indices and $\epsilon^{\mu \gamma \alpha \beta} \equiv \epsilon_{\mu \gamma \alpha \beta}$ is equal to 1 if the number of transpositions required to permute $(0123)$ to $(\mu\gamma\alpha\beta)$ is even; otherwise it takes the value -1.


Consider 2 different hyperlinks, $\oL = \{\ol^u:\ u=1, \ldots, \on\}$ and $\uL = \{\ul^v:\ v=1, \ldots, \un\}$. The former will be called a matter hyperlink; the latter will be referred to as a geometric hyperlink. The symbols $u, \bar{u}, v, \bar{v}$ will be indices, taking values in $\mathbb{N}$. They will keep track of the loops in our hyperlinks $\oL$ and $\uL$. The symbols $\on$ and $\un$ will always refer to the number of components in $\oL$ and $\uL$ respectively.

Color the matter hyperlink, i.e. choose a representation $\rho_u: \mathfrak{su}(2) \times \mathfrak{su}(2) \rightarrow {\rm End}(V_u^+) \times {\rm End}(V_u^-)$ for each component $\ol^u$, $u=1, \ldots, \on$, in the hyperlink $\oL$. Note that we do not color $\uL$, i.e. we do not choose a representation for $\uL$.

Let us return back to $\mathfrak{su}(2) \times \mathfrak{su}(2)$. Recall the first copy of $\mathfrak{su}(2)$ is generated by $\{\breve{e}_i\}_{i=1}^3$, which corresponds to boost in the $x_i$ direction in the Lorentz group; the second copy of $\mathfrak{su}(2)$ is generated by another independent set $\{\breve{e}_i\}_{i=1}^3$, which corresponds to rotation about the $x_i$-axis in the Lorentz group. When we give a representation $\rho^\pm$ to a colored loop $\ol$, which we interpret as representing a particle, we are effectively assigning values to the translational and angular momentum of this particle.

Given a colored hyperlink $\oL$ and a hyperlink $\uL$, we also assume that together (by using ambient isotopy if necessary), they form another hyperlink with $\on + \un$ components. Denote this new colored hyperlink by $\chi(\oL, \uL) \equiv \chi(\{\ol^u\}_{u=1}^{\on}, \{\ul^v\}_{v=1}^{\un})$, assumed to be time-like.

Let $q \in \bR$ be known as a charge. Define
\begin{align*}
V(\{\ul^v\}_{v=1}^{\underline{n}})(e) :=& \exp\left[ \sum_{v=1}^{\un} \int_{\ul^v} \sum_{\gamma=0}^3 B^i_\gamma \otimes dx_i\right], \\
W(q; \{\ol^u, \rho_u\}_{u=1}^{\on})(\omega) :=& \prod_{u=1}^{\on}\Tr_{\rho_u}  \mathcal{T} \exp\left[ q\int_{\ol^u} A^i_{\alpha\beta} \otimes dx_i\otimes \hat{E}^{\alpha\beta}  \right].
\end{align*}
Here, $\mathcal{T}$ is the time-ordering operator as defined in \cite{CS-Lim02}. And we sum over repeated indices, with $i$ taking values in 1, 2 and 3.

\begin{rem}
The term $\mathcal{T} \exp\left[ q\int_{\ol^u} \omega \right]$ is known as a holonomy operator along a loop $\ol^u$, for a spin connection $\omega$.
\end{rem}

Let $\vec{y}^u \equiv (y_0^u, y_1^u, y_2^u, y_3^u) : I := [0,1] \rightarrow \bR \times \bR^3$ be a parametrization of a loop $\ol^u \subset \oL$, $u=1, \ldots, \on$. We will write $y^u(s) = (y_1^u(s), y_2^u(s), y_3^u(s))$ and $\vec{y}^u(s) \equiv \vec{y}_s^u$. We will also write $\vec{y}^u = (y^u_0, y^u)$. Similarly, choose a parametrization $\vec{\varrho}^{v}: I \rightarrow \bR \times \bR^3$, $v = 1, \ldots, \un$, for each loop $\ul^v \subset \uL$. When the loop is oriented, we will often choose a parametrization which is consistent with the assigned orientation.

Fix a closed and bounded orientable surface, with or without boundary, denoted by $S$, inside $\bR^3$. We can identify $S$ as $\{0\} \times S$ inside $\bR \times \bR^3$ and write $S$ as a finite disjoint union of surfaces. Because in Loop Quantum Gravity, we consider topological equivalent surfaces, thus without loss of generality, we may and will assume that $S$ lies inside the $x_2-x_3$ plane.
Furthermore, we insist the hyperlink $\oL$ is disjoint from $\{0\} \times S$ and $\pi_0(\oL)$ intersect at most finitely many points inside $S$.

Using the dynamical variables $\{B_j\gamma^i\}$ and the Minkowski metric $\eta^{ab}$, we see that the metric $g^{ab} \equiv B^a_\mu\eta^{\mu\gamma}B^b_\gamma$ and the
area is given by \beq {\rm Area\ of}\ S(e) \equiv A_S(e) := \int_S \sqrt{g^{22}g^{33} - (g^{23})^2} \ dA. \nonumber \eeq

Consider the following path integral, \beq \frac{1}{Z}\int_{\omega \in L_\omega,\ e \in L_e}V(\{\ul^v\}_{v=1}^{\un})(e) W(q; \{\ol^u, \rho_u\}_{u=1}^{\on})(\omega)A_S(e)\  e^{i S_{EH}(e, \omega)}\ De D\omega, \label{e.eha.1} \eeq whereby $De$ and $D\omega$ are Lebesgue measures on $L_e$ and $L_\omega$ respectively and \beq Z = \int_{\omega \in L_\omega,\ e \in L_e}e^{i S_{EH}(e, \omega)}\ De D\omega,\ i=\sqrt{-1}. \nonumber \eeq

\begin{rem}\label{r.sux.1}
\begin{enumerate}
\item When $S$ is the empty set, we define $A_\emptyset \equiv 1$, so we write Expression \ref{e.eha.1} as $Z(q; \chi(\oL, \uL))$, which in future be termed as the Wilson Loop observable of the colored hyperlink $\chi(\oL, \uL))$.
  \item We will write Expression \ref{e.eha.1} as $\hat{A}_S[Z(q; \chi(\oL, \uL))]$, whereby $\hat{A}_S$ is the quantized area operator for surface $S$.
\end{enumerate}
\end{rem}

\begin{notation}\label{n.n.4}
Suppose we have a list of irreducible representations of $\mathfrak{su}(2)\times \mathfrak{su}(2)$, $\{\rho_1, \ldots, \rho_{\on}\}$ and there are $\bar{m}$ distinct representations, labeled as $\{\gamma_1, \ldots, \gamma_{\bar{m}} \}$, arranged in any order. For representation $\gamma_u$, let $\Gamma_u$ denote the set of integers in $\{1,2, \ldots, \on\}$ such that $\rho_v = \gamma_u$, $v \in \Gamma_u$.
\end{notation}

Write the surface $S$ into a disjoint union of smaller surfaces \beq S_1,\ S_2,\ \ldots,\ S_{\bar{m}}, \nonumber \eeq $\bar{m}$ as defined in Notation \ref{n.n.4}. In other words, $S_v$ will be the (possibly disconnected) surface whereby those hyperlinks colored with the same representation $\gamma_v$ pierce it.
Let $I^2 \equiv I \times I = \bigcup_{v=1}^{\bar{m}}I_v^2$ such that $\sigma: I_v^2 \rightarrow S_v$ is a parametrization of $S_v$.

In \cite{EH-Lim02}, we showed that we can define Expression \ref{e.eha.1} and write it as the limit as $\kappa$ goes to infinity, of the following expression
\begin{align}
\prod_{\bar{u}=1}^{\on}\Bigg\{& \Bigg[ \sum_{v=1}^{\bar{m}} \bk^{3}\int_{I_v^2}\Bigg|\left\langle q \sum_{u \in \Gamma_v} \kappa\int_0^1 y_{1,s}^{u,\prime} \partial_0^{-1}p_\kappa^{\vec{y}_s^u} ds, p_\kappa^{\vec{\sigma}(\hat{t})} \right\rangle\ \sqrt{\xi_{\rho_u^+}}
\Bigg|\ J_{23}(\hat{t}) d\hat{t} \Bigg]^{1/\on} \Tr_{\rho_{\bar{u}}^+}\ \hat{\mathcal{W}}_\kappa^+(q; \ol^{\bar{u}}, \uL) \nonumber\\
+& \Bigg[
i\sum_{v=1}^{\bar{m}}\bk^{3}\int_{I_v^2}\Bigg|\left\langle q \sum_{u \in \Gamma_v} \kappa\int_0^1 y_{1,s}^{u,\prime} \partial_0^{-1}p_\kappa^{\vec{y}_s^u} ds, p_\kappa^{\vec{\sigma}(\hat{t})} \right\rangle\ \sqrt{\xi_{\rho_u^-}}
\Bigg|\ J_{23}(\hat{t}) d\hat{t} \Bigg]^{1/\on} \Tr_{\rho_{\bar{u}}^-}\ \hat{\mathcal{W}}_\kappa^-(q; \ol^{\bar{u}}, \uL) \Bigg\}
. \label{e.a.10}
\end{align}

Note that $i = \sqrt{-1}$ and $\bk := \kappa/2\sqrt{4\pi}$. See Definition Area Path Integral in \cite{EH-Lim02}. We will postpone the explanation of this expression to Section \ref{s.ne}.

\begin{rem}
Expression \ref{e.a.10} is dependent on how we partition the surface $S$ into $\bigcup_{v=1}^{\bar{m}}S_v$. But its limit as $\kappa$ goes to infinity will be shown to be independent of this partition.
\end{rem}


From the main Theorem \ref{t.main.2} in this article, we will show that the path integral can be computed from the number of times the link $\pi_0(\oL)$ pierce the surface $S$, weighted by the momentum corresponding to the representation $\rho_u$. The momentum is given by a 2-component vector; the first component corresponds to translational momentum, the second component corresponds to angular momentum.

It was remarked in \cite{rovelli1995discreteness} that the eigenvalues for the area operator come from matter in the ambient space. This is consistent with our main result, whereby the eigenvalues of the area operator are computed from the matter hyperlink $\oL$.

We also wish to add that one can also interpret the eigenvalues of the area operator as quanta of area, which gives rise to the notion of quantum geometry. For the details, we refer the reader to \cite{Mercuri:2010xz} and \cite{rovelli2004quantum}. The fact that the eigenvalues are discrete implies the discreteness of quantum geometry. See \cite{0264-9381-21-15-R01}, \cite{rovelli2004quantum}, \cite{Rovelli1998} and \cite{Thiemann:2007zz}.

\section{Important Notations}\label{s.ne}

Let us now explain Expression \ref{e.a.10}. The notations and details are all taken from \cite{EH-Lim02}, which we have condensed into this section.

In this article, $\vec{y} \equiv (y_0, y) \in \bR^4$, whereby $y \equiv (y_1, y_2, y_3) \in \bR^3$. We will write
\beq \hat{y}_i =
 \left\{
  \begin{array}{ll}
    (y_2, y_3), & \hbox{$i=1$;} \\
    (y_1, y_3), & \hbox{$i=2$;} \\
    (y_1, y_2), & \hbox{$i=3$.}
  \end{array}
\right. \nonumber \eeq

If $x \in \bR^n$, we will write $(p_\kappa^x)^2$ to denote the $n$-dimensional Gaussian function, center at $x$, variance $1/\kappa^2$. For example, \beq p_\kappa^x(\cdot) = \frac{\kappa^2}{2\pi}e^{-\kappa^2|\cdot - x|^2/4},\ x \in \bR^4. \nonumber \eeq We will also write $(q_\kappa^x)^2$ to denote the 1-dimensional Gaussian function, i.e. \beq q_\kappa^x(\cdot) = \frac{\sqrt{\kappa}}{(2\pi)^{1/4}}e^{-\kappa^2 (\cdot - x)^2/4}. \nonumber \eeq

For $x, y \in \bR^2$, we write \beq \langle p_\kappa^x, p_\kappa^y \rangle = \int_{z\in \bR^2} \frac{\kappa}{\sqrt{2\pi}}e^{-\kappa^2|z - x|^2/4} \frac{\kappa}{\sqrt{2\pi}}e^{-\kappa^2|z - y|^2/4} dz, \nonumber \eeq i.e. we integrate over Lebesgue measure on $\bR^2$. More generally, given $f, g \in C(\bR^n)$, we will write \beq \langle f, g \rangle \equiv \int_{\bR^n} f\cdot g d\lambda, \nonumber \eeq whereby $\lambda$ is Lebesgue measure.

For $x = (x_0,x_1, x_2, x_3)$, write \beq x(s_a) :=
\left\{
  \begin{array}{ll}
    (s_0,x_1, x_2, x_3), & \hbox{$a=0$;} \\
    (x_0,s_1, x_2, x_3), & \hbox{$a=1$;} \\
    (x_0,x_1, s_2, x_3), & \hbox{$a=2$;} \\
    (x_0,x_1, x_2, s_3), & \hbox{$a=3$.}
  \end{array}
\right. \nonumber \eeq

Let $\partial_a \equiv \partial/\partial x_a$ be a differential operator. There is an operator $\partial_a^{-1}$ acting on a dense subset in $\overline{\mathcal{S}}_\kappa(\bR^4)$, \beq (\partial_a^{-1}f)(x) := \frac{1}{2}\int_{-\infty}^{x_a} f(x(s_a))\ ds_a - \frac{1}{2}\int_{x_a}^{\infty} f(x(s_a))\ ds_a,\ f \in \overline{\mathcal{S}}_\kappa(\bR^4). \label{e.d.1} \eeq Here, $x_a \in \bR$. Notice that $\partial_a\partial_a^{-1}f \equiv f$ and $\partial_a^{-1}f$ is well-defined provided $f$ is in $L^1$.

For each $i = 1, 2, 3$, write
\beq \left\langle p_\kappa^{\vec{x}}, p_\kappa^{\vec{y}} \right\rangle_i =
\left\langle p_\kappa^{\hat{x}_{i}}, p_\kappa^{\hat{y}_{i}} \right\rangle \left\langle q_\kappa^{x_{i}}, \kappa\partial_0^{-1}q_\kappa^{y_{i}} \right\rangle\left\langle \partial_0^{-1}q_\kappa^{x_{0}}, q_\kappa^{y_{0}}\right\rangle. \nonumber \eeq

Here, \beq \partial_0^{-1}q_\kappa^{x_{0}}(t) \equiv \frac{1}{2}\int_{-\infty}^t q_\kappa^{x_{0}}(\tau)\ d\tau -
\frac{1}{2}\int_{t}^\infty q_\kappa^{x_{0}}(\tau)\ d\tau. \nonumber \eeq

Note that $\left\langle \partial_0^{-1}q_\kappa^{x_{0}}, q_\kappa^{y_{0}}\right\rangle \equiv \left\langle q_\kappa^{y_{0}}, \partial_0^{-1}q_\kappa^{x_{0}} \right\rangle$  means we integrate $\partial_0^{-1}q_\kappa^{x_{0}} \cdot q_\kappa^{y_{0}}$ over $\bR$, using Lebesgue measure. It is well-defined because $q_\kappa^{x_0}$ is in $L^1$.

Recall we parametrize $\ol^u$ and $\ul^v$ using $\vec{y}^u$ and $\vec{\varrho}^v$ respectively, $u=1, \ldots \on$, $v=1, \ldots \un$. Define $\hat{\mathcal{W}}_\kappa^\pm(q; \ol^u, \uL)$ as
\begin{align}
\hat{\mathcal{W}}_\kappa^\pm(q; \ol^u, \uL) := \exp\left[ \mp\frac{iq}{4}\frac{\kappa^3}{4\pi}\sum_{v=1}^{\underline{n}}\int_{I^2}\ d\hat{s}\ \epsilon^{ijk}\left\langle  p_\kappa^{\vec{y}^u_s}, p_\kappa^{\vec{\varrho}^v_{\bar{s}}}\right\rangle_k  y^{u,\prime}_{i,s}\varrho^{v,\prime}_{j,\bar{s}} \otimes \mathcal{E}^\pm\right], \label{e.h.5}
\end{align}
whereby $\mathcal{E}^\pm$ was defined in Equation (\ref{e.sux.1}). And $\epsilon^{ijk} \equiv \epsilon_{ijk}$ be defined on the set $\{1,2,3\}$, by \beq \epsilon^{123} = \epsilon^{231} = \epsilon^{312} = 1,\ \ \epsilon^{213} = \epsilon^{321} = \epsilon^{132} = -1, \nonumber \eeq if $i,j,k$ are all distinct; 0 otherwise. Finally $\Tr_{\rho_u^\pm}$ means take the trace and $J_{23}$ is defined in Notation \ref{n.s.3}.

\section{Area operator}\label{s.ao}

We will now discuss computing the limit as $\kappa$ goes to infinity,
\begin{align*}
\frac{\kappa^3}{16\pi}&\int_{I^2}\ d\hat{s}\ \left\langle  p_\kappa^{\vec{y}^u_s}, p_\kappa^{\vec{\varrho}^v_{\bar{s}}}\right\rangle_k  y^{u,\prime}_{i,s}\varrho^{v,\prime}_{j,\bar{s}} \\
\equiv& \frac{\kappa^3}{16\pi}\int_0^1\int_0^1 ds d\bar{s}\ \left\langle p_\kappa^{\hat{y}_{i,s}}, p_\kappa^{\hat{\varrho}_{i,\bar{s}}} \right\rangle \left\langle q_\kappa^{y_{i,s}}, \kappa\partial_i^{-1}q_\kappa^{\varrho_{i,\bar{s}}} \right\rangle \left\langle \partial_0^{-1}q_\kappa^{y_{0,s}}, q_\kappa^{\varrho_{0,\bar{s}}} \right\rangle[y_{j,s}'\varrho_{k,\bar{s}}' - y_{k,s}'\varrho_{j,\bar{s}}'],
\end{align*}
for $(i,j,k) \in \{(1,2,3), (2,3,1), (3,1,2)\}$. Given a link, which is a set of disjoint curves $L = \{l^1, l^2, \ldots, l^{\hat{m}}\} \subset \bR^3$, we will need to project it onto a plane in a `nice' manner, to form a link diagram.
Most, if not all, of our calculations can be obtained from a link diagram.

Definition \ref{d.l.1}, is actually taken from
Definition 4.3 in \cite{CS-Lim01}. We have modified the definition to accommodate a hyperlink in $\bR^4$, instead of a link inside $\bR^3$, as in Definition 4.3 in \cite{CS-Lim01}.

Fix $i = 1, 2, 3$ and let $(i,j,k) \in \{(1,2,3), (2,3,1), (3,1,2)\}$. Consider 2 distinct parametrized curves, $y, \varrho: [0,1] \rightarrow \bR^3$. Now, the term \beq \lim_{\kappa \rightarrow \infty}\frac{\kappa^2}{4\pi\sqrt{8\pi}}\int_0^1\int_0^1 ds d\bar{s}\ \left\langle p_\kappa^{\hat{y}_{i,s}}, p_\kappa^{\hat{\varrho}_{i,\bar{s}}} \right\rangle \left\langle q_\kappa^{y_{i,s}}, \kappa\partial_i^{-1}q_\kappa^{\varrho_{i,\bar{s}}} \right\rangle [y_{j,s}'\varrho_{k,\bar{s}}' - y_{k,s}'\varrho_{j,\bar{s}}'] \nonumber \eeq
calculates the linking number between the 2 oriented curves, when the 2 curves are projected `nicely' on the plane $\Sigma_i$, as in Definition \ref{d.l.1}. We will not prove this result here, but instead refer the reader to \cite{CS-Lim01} for a detailed discussion. The linking number between the 2 oriented curves is the sum of all the algebraic crossing number of the crossings formed by the 2 curves on a link diagram. Since the linking number between 2 oriented curves is well-defined and is invariant under ambient isotopy, it does not matter which plane we choose to project the 2 curves on. Hence, we see that the double integral gives the same value for each $i= 1, 2, 3$, corresponding to the projection onto the $\Sigma_i$ plane respectively, defined in Notation \ref{n.s.1}.

However, notice that we need to consider an additional term $\kappa\left\langle \partial_0^{-1}q_\kappa^{x_{0}}, q_\kappa^{y_{0}}\right\rangle/\sqrt{2\pi}$, which arises because we are dealing with curves in $\bR \times \bR^3$. This term, when we take the limit as $\kappa$ goes to infinity, gives us $+ 1$ when $y_0$ lies above $x_0$; -1 otherwise.

From Lemma \ref{l.l.1}, we have the following corollary.

\begin{cor}(Hyperlinking number between $\ol$ and $\ul$.)\label{c.c.1}
Let $\ol$ and $\ul$ be 2 distinct oriented loops and $\vec{y}$ and $\vec{\varrho}$ are its respective parametrizations, and together they form a new oriented hyperlink $\chi(\ol, \ul)$. Refer to Definition \ref{d.l.1}. Then,
\begin{align}
\lim_{\kappa \rightarrow \infty}\frac{\kappa^3}{16\pi^2} &\int_{I^2}d\hat{s}\ \epsilon^{ijk}\left\langle  p_\kappa^{\vec{y}_s}, p_\kappa^{\vec{\varrho}_{\bar{s}}}\right\rangle_k y_{i,s}^{\prime}\varrho_{j,\bar{s}}^{\prime} \nonumber \\
=& \sum_{k=1}^3\sum_{p \in \pd(\Sigma_k; \hat{y}_k, \hat{\varrho}_k)}\sgn(p;\hat{y}_k: \hat{\varrho}_k)\cdot\sgn(p;y_k : \varrho_k)\cdot \sgn(p;y_0 : \varrho_0) =: {\rm sk}(\ol, \ul).\nonumber
\end{align}
\end{cor}

\begin{rem}
The quantity ${\rm sk}(\ol, \ul)$ is not exactly the linking number between the 2 curves projected in $\bR^3$. This is because we need to take into account the time-lag of the crossings, which only appear for curves in $\bR \times \bR^3$. We will refer ${\rm sk}(\ol, \ul)$ as the hyperlinking number between $\ol$ and $\ul$ in $\bR \times\bR^3$, to distinguish from the linking number between 2 simple closed curves in $\bR^3$. Of course, the hyperlinking number is calculated from the hyperlink $\chi(\ol, \ul)$.
\end{rem}

Recall from Remark \ref{r.sux.1} that we define the Wilson Loop observable of the colored hyperlink $\chi(\ol, \ul)$ as Expression \ref{e.eha.1}, with $S = \emptyset$. Thus, Expression \ref{e.a.10} will give us
\beq  Z(q; \chi(\oL, \uL) ) := \lim_{\kappa \rightarrow \infty}\prod_{u=1}^{\on}\left[\Tr_{\rho_u^+}\ \hat{\mathcal{W}}_\kappa^+(q; \ol^u, \uL) + \Tr_{\rho_u^-}\ \hat{\mathcal{W}}_\kappa^-(q; \ol^u, \uL) \right]. \nonumber \eeq See \cite{EH-Lim02}.

\begin{thm}(Wilson Loop observable)\label{t.w.1}
Consider two oriented time-like hyperlinks, $\oL = \{\ol^u \}_{u=1}^{\on}$, $\uL = \{\ul^v \}_{v=1}^{\un}$ in $\bR \times \bR^3$ with non-intersecting (closed) loops, the former colored with a representation $\rho_u \equiv (\rho^+_u, \rho^-_u): \mathfrak{su}(2)\times \mathfrak{su}(2) \rightarrow {\rm End}(V_u^+) \times {\rm End}(V_u^-)$. From these 2 oriented hyperlinks, form a new colored oriented time-like hyperlink, denoted by $\chi(\oL, \uL)$.

Recall from Equation (\ref{e.sux.1}) that $\mathcal{E}^\pm = \sum_{i=1}^3\breve{e}_i$ and $\mathcal{E}$ from Equation (\ref{e.sux.2}). Define for each $u=1, \ldots, \on$, \beq {\rm sk}(\ol^u, \uL):= \sum_{v=1}^{\un}{\rm sk}(\ol^u, \ul^v), \nonumber\eeq calculated from $\chi(\oL, \uL)$. Furthermore, we have
\begin{align}
\prod_{u=1}^{\on}\Tr_{\rho_u}\exp &\left[ \frac{iq}{4}\frac{\kappa^3}{4\pi}\sum_{v=1}^{\un}\int_{I^2}\ d\hat{s}\ \epsilon^{ijk}\left\langle  p_\kappa^{\vec{y}^u_s}, p_\kappa^{\vec{\varrho}^v_{\bar{s}}}\right\rangle_k  y^{u,\prime}_{i,s}\varrho^{v,\prime}_{j,\bar{s}} \otimes \mathcal{E}\right]\nonumber \\
& \longrightarrow_{\kappa \rightarrow \infty}\prod_{u=1}^{\on}\Tr_{\rho_u}\ \exp[\pi iq\ {\rm sk}(\ol^u, \uL) \cdot \mathcal{E}] \nonumber \\
&\equiv \prod_{u=1}^{\on}\left( \Tr_{\rho^+_u}\ \exp[-\pi iq\ {\rm sk}(\ol^u, \uL) \cdot \mathcal{E}^+] +
\Tr_{\rho^-_u}\ \exp[\pi iq\ {\rm sk}(\ol^u, \uL) \cdot \mathcal{E}^-] \right). \label{e.h.4}
\end{align}

Hence the Wilson Loop observable can be computed as
\begin{align*}
Z(q; & \chi(\oL, \uL) )
= \prod_{u=1}^{\on}\left( \Tr_{\rho^+_u}\ \exp[\pi iq\ {\rm sk}(\ol^u, \uL) \cdot \mathcal{E}^+] +
\Tr_{\rho^-_u}\ \exp[-\pi iq\ {\rm sk}(\ol^u, \uL) \cdot \mathcal{E}^-] \right).
\end{align*}
\end{thm}

\begin{proof}
From Corollary \ref{c.c.1}, \beq \frac{iq}{4}\frac{\kappa^3}{4\pi}\sum_{v=1}^{\on}\int_{I^2}\ d\hat{s}\ \epsilon^{ijk}\left\langle  p_\kappa^{\vec{y}^u_s}, p_\kappa^{\vec{\varrho}^v_{\bar{s}}}\right\rangle_k  y^{u,\prime}_{i,s}\varrho^{v,\prime}_{j,\bar{s}} \nonumber \eeq converges to $\pi\ {\rm sk}(\ol^u, \uL)$. This says that \beq  \lim_{\kappa \rightarrow \infty}\Tr_{\rho_u^\pm}\ \hat{\mathcal{W}}_\kappa^\pm(q; \ol^u, \uL) = \Tr_{\rho^\pm_u}\ \exp[\mp\pi iq\ {\rm sk}(\ol^u, \uL) \cdot \mathcal{E}^\pm]. \label{e.w.5} \eeq So the path integral converges to $\prod_{u=1}^{\on}\Tr_{\rho_u}\ \exp \left[\pi iq\ {\rm sk}(\ol^u, \uL) \mathcal{E} \right]$.
\end{proof}

We can generalize the above theorem, when the surface $S$ is not empty. To state the general theorem, we need the following lemma.

\begin{lem}\label{l.s.1}
Identify $S \subset \bR^3$ with $\{0\} \times S \subset \bR \times \bR^3$ and by abuse of notation, denote it by $S$. Refer to Definition \ref{d.s.1}. We have \beq \lim_{\kappa \rightarrow \infty}\frac{\kappa^3}{32\pi^2}\int_{I^2}\left[ \int_0^1\ ds\ \kappa\left\langle \partial_0^{-1}p_\kappa^{\vec{y}_s^u}, p_\kappa^{\vec{\sigma}(\hat{t})} \right\rangle\ y_{1,s}^{u,\prime} \right] J_{23}(\hat{t}) d\hat{t} = {\rm lk}(\ol^u, S). \nonumber \eeq
\end{lem}

\begin{proof}
Observe that we can write the surface area $S$ into a finite number of disjoint surfaces. So it suffices to consider each such surface whereby $l\equiv l^u$ pierces it at most once.

Now, we can write \beq \left\langle \partial_0^{-1}p_\kappa^{\vec{y}_s^u}, p_\kappa^{\vec{\sigma}(\hat{t})} \right\rangle = \left\langle p_\kappa^{y_s^u}, p_\kappa^{\sigma(\hat{t})}\right\rangle \cdot \left\langle \partial_0^{-1}q_\kappa^{y_{0,s}^u}, q_\kappa^0 \right\rangle. \nonumber \eeq

Using the lemma found in the appendix in \cite{EH-Lim02}, we have
\begin{align*}
\frac{\kappa^3}{32\pi^2}\int_{I^2} \int_0^1 & \left\langle p_\kappa^{y_s^u}, p_\kappa^{\sigma(\hat{t})}\right\rangle \cdot \left\langle \kappa\partial_0^{-1}q_\kappa^{y_{0,s}^u}, q_\kappa^0 \right\rangle\ y_{1,s}^{u,\prime} J_{23}(\hat{t}) ds d\hat{t} \\
&= \frac{\kappa^3}{32\pi^2}\int_{I^3}  e^{-\kappa^2|y_s^u - \sigma(\hat{t})|^2/8} \left\langle \kappa\partial_0^{-1}q_\kappa^{y_{0,s}^u} , q_\kappa^{0}  \right\rangle y_{1,s}^{u,\prime}\ J_{23}(\hat{t}) d\hat{t} ds,
\end{align*}
if we write $I^3 = I^2 \times [0,1]$.

Make the substitution $(s, \hat{t}) \longmapsto \omega(s, \hat{t}) = y_s^u - \sigma(\hat{t})$ and $z(\omega) := q_\kappa^{y_{0,s}^u}$. If $l$ intersects $S$, then $0 \in \omega(I^3)$.

Then the integral becomes a volume integral \beq {\rm sgn}\left(y_{1}^{u,\prime} J_{23} \right)\frac{\kappa^3}{16\pi\sqrt{2\pi}}\int_{\omega(I^3)} e^{-\kappa^2|\omega |/8} \left\langle \frac{\kappa}{\sqrt{2\pi}}\partial_0^{-1}q_\kappa^{z(\omega)} , q_\kappa^{0}  \right\rangle\ d\omega. \nonumber \eeq Since
\begin{align*}
\frac{\kappa^3}{8(2\pi)^{3/2}}\int_{\omega(I^3)} e^{-\kappa^2|\omega |/8} \ d\omega \rightarrow&
\left\{
  \begin{array}{ll}
    1, & \hbox{$0 \in \omega(I^3)$;} \\
    0, & \hbox{$0 \notin \omega(I^3)$,}
  \end{array}
\right.
 \\
\frac{1}{\sqrt{2\pi}}\left\langle\kappa\partial_0^{-1}q_\kappa^{z(\omega)} , q_\kappa^{0}  \right\rangle \rightarrow& \pm 1,
\end{align*}
then we have
\begin{align*}
\frac{\kappa^3}{32\pi^2}\int_{\omega(I^3)} e^{-\kappa^2|\omega |/8} \left\langle \kappa\partial_0^{-1}q_\kappa^{z(\omega)} , q_\kappa^{0}  \right\rangle\ d\omega \rightarrow
\left\{
  \begin{array}{ll}
   {\rm lk}(l, S), & \hbox{$0 \in \omega(I^3)$;} \\
    0, & \hbox{$0 \notin \omega(I^3)$.}
  \end{array}
\right.
\end{align*}
\end{proof}

We can now state our main theorem.

\begin{thm}(Main Theorem)\label{t.main.2}\\
Consider two oriented hyperlinks, $\oL = \{\ol^u \}_{u=1}^{\on}$, $\uL = \{\ul^v \}_{v=1}^{\un}$ in $\bR \times \bR^3$ with non-intersecting (closed) loops, the former colored with a representation $\rho_u \equiv (\rho^+_u, \rho^-_u): \mathfrak{su}(2)\times \mathfrak{su}(2) \rightarrow {\rm End}(V_u^+) \times {\rm End}(V_u^-)$. From these 2 oriented hyperlinks, form a new colored oriented hyperlink, denoted by $\chi(\oL, \uL)$.

Let $\hat{A}_S$ be the area operator corresponding to a surface $S \subset \bR^3$, which is embedded as a surface inside $\bR \times \bR^3$, as in Definition \ref{d.s.1}. By abuse of notation, we still write $S \equiv \{0\} \times S \subset \bR \times \bR^3$ and we assume that $\oL$ is disjoint from $S$ and $\pi_0(\oL)$ intersect $S$ at most finitely many points. From Definition \ref{d.s.1}, let $\pd(\pi_0; \ol^u, S)$ denote the set of piercings between $\pi_0(\ol^u)$ and $S$ and $\left|\pd(\pi_0; \ol^u, S) \right|$ denote the number of piercings inside the set.

Then $\hat{A}_S$ acts on the Wilson Loop observable of the colored hyperlink $\chi(\oL, \uL)$, $Z(q; \chi(\oL, \uL))$ via the path integral Expression \ref{e.eha.1}, which is defined as the limit of Expression \ref{e.a.10} as $\kappa$ goes to infinity, given by
\begin{align*}
\hat{A}_S\left[Z(q; \chi(\oL, \uL))\right] :=\frac{|q|\sqrt\pi}{2}\prod_{\bar{u}=1}^{\on} \Bigg\{& \left[
\sum_{u=1}^{\on}\left|\pd(\pi_0; \ol^u, S) \right|  \sqrt{\xi_{\rho_u^+}}\right]^{1/\on}
\Tr_{\rho^+_{\bar{u}}}\ \exp[-\pi iq\ {\rm sk}(\ol^{\bar{u}}, \uL) \cdot \mathcal{E}^+]
\\
+&
\left[ i\sum_{u=1}^{\on} \left|\pd(\pi_0; \ol^u, S)\right| \sqrt{\xi_{\rho_u^-}} \right]^{1/\on}
\Tr_{\rho^-_{\bar{u}}}\ \exp[\pi iq\ {\rm sk}(\ol^{\bar{u}}, \uL) \cdot \mathcal{E}^-]\ \Bigg\}
.
\end{align*}
\end{thm}

\begin{proof}
We need to compute the limit as $\kappa$ goes to infinity, of \beq
\sum_{v=1}^{\bar{m}} \bk^{3}\int_{I_v^2}\Bigg|\left\langle q \sum_{u \in \Gamma_v} \kappa\int_0^1 y_{1,s}^{u,\prime} \partial_0^{-1}p_\kappa^{\vec{y}_s^u} ds, p_\kappa^{\vec{\sigma}(\hat{t})} \right\rangle\ \sqrt{\xi_{\rho_u^+}}
\Bigg|\ J_{23}(\hat{t}) d\hat{t} . \nonumber \eeq

Observe that the integral reduces to a finite sum of terms in the limit, each term is either $\pm 1$, by Lemma \ref{l.s.1}. Thus
\begin{align*}
\frac{\kappa^3}{32\pi^2}&\int_{\hat{t} \in I_v^2}\left| \sum_{u\in \Gamma_v^\pm} \int_{s \in I} \kappa\left\langle \partial_0^{-1}p_\kappa^{\vec{y}_s^u}, p_\kappa^{\vec{\sigma}(\hat{t})} \right\rangle\ y_{1,s}^{u,\prime} \right|\ J_{23}(\hat{t}) d\hat{t}
\longrightarrow
\sum_{u\in \Gamma_v^\pm} \left|\pd(\pi_0; \ol^u, S) \right|
\end{align*}
as $\kappa$ goes to infinity and together with Equation (\ref{e.w.5}), the result now follows from plugging the above limits inside Expression \ref{e.a.10}.
\end{proof}

\begin{rem}
Note that the area functional $A_S$ commutes with the holonomy operator. As such, we expect and require that the quantized operator $\hat{A}_S$ be proportional to the identity.
\end{rem}

Observe that the area operator is made up of a real part and an imaginary part from the main Theorem \ref{t.main.2}. The term $\sqrt{\xi_{\rho_u^+}}$ represents the translational momentum coming from boost; the term $\sqrt{\xi_{\rho_u^-}}$ represents angular momentum coming from rotation. Hence the real part corresponds to momentum coming from boost and the imaginary part corresponds to angular momentum. Because of $\sqrt{-1}$, we see that both momentum do not cancel each other out and both should be considered together as a 2-component momentum vector.

Thus the area operator counts the total momentum of the particles colliding on a surface $S$ and it takes into account both the momentum coming from boost and rotation and we are allowed to know the total momentum coming from each.

\appendix

\section{Link Diagrams}

\begin{defn}\label{d.l.1}(Link Diagrams in $\Sigma_i$)\\
Assume that a finite set of simple closed curves $L = \{l^1,l^2, \ldots l^{\hat{m}}\} \in \mathbb{R}^4$ is $C^1$ and oriented. Furthermore, assume that $L$ is a time-like hyperlink.

Parametrise each curve, $\vec{y}^u = (y^u_0, y^u_1, y^u_2, y^u_3) \equiv (y^u_0, y^u ): [0,1] \rightarrow \mathbb{R}^4$ such that $|\vec{y}^{u,\prime}| \neq 0$. Recall we define $\hat{y}_i^u$ in Section $i=1,2,3$.

\begin{enumerate}
  \item Let $\pi_0: \bR \times \mathbb{R}^3 \rightarrow \mathbb{R}^3$ be the projection on the hyperplane $\bR^3$, whose normal is given by $(1,0,0,0)$. And let $\pi_i: \bR \times \bR^3 \rightarrow \bR \times \Sigma_i$ denote a projection. By assumption, $\pi_0(L)$ defines a link in $\bR^3$. Similarly, $\pi_i(L)$ defines a link in $\bR \times \Sigma_i$.

      Fix a $i=1,2,3$. Define a standard projection $P_i$ of the curve parametrized by $y^u \equiv \pi_0( \vec{y}^u) \subset \bR^3$ or $\pi_i(\vec{y}^u)$ onto the plane $\Sigma_i$, if the following conditions are satisfied:
      \begin{description}
        \item[a]\noindent for any point $p \in \Sigma_i$, $P_i^{-1}(p)$ intersects at most 2 distinct arcs in $L$. Now $p$ is called a crossing if $P_i^{-1}(p)$ intersects exactly 2 distinct arcs.
        \item[b]\noindent at each crossing $p = \hat{y}_i^u(s_0) = \hat{y}_i^v(t_0)$, there exists an $\varepsilon > 0$ such that for all $|s-s_0| < \varepsilon$ and $|t - t_0|< \varepsilon$, and $[y^{u,\prime}(s) \times y^{v,\prime}(t)]_i\neq 0$, i.e. the 2 projected arcs on the plane $\Sigma_i$ at the crossing $p$ are linearly independent. Here, $v_i$ refers to the $i$-th component of a 3-vector $v$ in $\bR^3$.
      \end{description}
       Denote the set of crossings in $\Sigma_i$ between curves $\vec{y}^u$ and $\vec{y}^v$ by $\pd(\Sigma_i;\ \hat{y}_i^u, \hat{y}_i^v)$, $u \neq v$. 
  \item For each curve parametrized by $\hat{y}_i^u$, write the interval $[0,1]$ as a union of intervals \\ $\bigcup_{a=1}^{n(\hat{y}_i^u)} A(\hat{y}_i^u)^a$, where in each interval $A(\hat{y}_i^u)^a$, $s \in A(\hat{y}_i^u)^a \mapsto \hat{y}_i^u(s)$ is a bijection, for each $i=1,2,3$. Write $C(\hat{y}_i^u)^a := \hat{y}_i^u(A(\hat{y}_i^u)^a) \in \Sigma_i$ be the image of the interval $A(\hat{y}_i^u)^a$ under $\hat{y}_i^u$. Without loss of generality, further assume each $C(\hat{y}_i^u)^a$ contains at most one crossing which is an interior point in $C(\hat{y}_i^u)^a$.
  \item Given 2 arcs $C(\hat{y}_i^u)^a, C(\hat{y}_i^v)^{\hat{a}}$ which intersect at $p \in \Sigma_i$, define $\sgn(J(C(\hat{y}_i^u)^a, C(\hat{y}_i^v)^{\hat{a}}))$ to be the sign of the determinant of the Jacobian $J(C(\hat{y}_i^u)^a, C(\hat{y}_i^v)^{\hat{a}})= [y^{u,\prime}(s) \times y^{v,\prime}(t)]_i$ at the crossing $p = \hat{y}_i^u(s_0) =\hat{y}_i^v(t_0)$. Otherwise, define it to be zero if the 2 arcs do not intersect at all. We will also write $\sgn(p;\hat{y}_i^u:\hat{y}_i^v) \equiv \sgn(J(C(\hat{y}_i^u)^a, C(\hat{y}_i^v)^{\hat{a}}))$, $p \in C(\hat{y}_i^u)^a\cap C(\hat{y}_i^v)^{\hat{a}}$ and call this the orientation of $p$.
  \item Using the same notation as the previous item, for each crossing $p \in C(\hat{y}_i^u)^a \cap C(\hat{y}_i^v)^{\hat{a}}$, define
  \beq \sgn(C(\hat{y}_i^u)^a:C(\hat{y}_i^v)^{\hat{a}}) = \left\{
                                  \begin{array}{ll}
                                    1, & \hbox{$y_i^u > y_i^v$;} \\
                                    -1, & \hbox{$y_i^u < y_i^v$.}
                                  \end{array}
                                \right. \nonumber \eeq
If the 2 arcs do not intersect, set it to be 0. We will also write $\sgn(p;y_i^u:y_i^v) \equiv \sgn(C(\hat{y}_i^u)^a:C(\hat{y}_i^v)^{\hat{a}})$ and call this the height of $p$.

\item For each crossing $p \in \pd(\Sigma_i; \hat{y}_i^u, \hat{y}_i^v)$, the algebraic crossing number is defined by \begin{align} \varepsilon(p) =& \sgn(p;\hat{y}_i^u:\hat{y}_i^v)\cdot\sgn(p;y_i^u:y_i^v) \in \{\pm 1\}. \nonumber \end{align}
This is actually well defined on an oriented link diagram in $\Sigma_i$, independent of the parametrization used.

\item Using the same notation as the previous item, for each crossing $p \in C(\hat{y}_i^u)^a \cap C(\hat{y}_i^v)^{\hat{a}}$, define
  \beq {\rm gap}(C(\hat{y}_i^u)^a:C(\hat{y}_i^v)^{\hat{a}}) = \left\{
                                  \begin{array}{ll}
                                    1, & \hbox{$y_0^u < y_0^v$;} \\
                                    -1, & \hbox{$y_0^u > y_0^v$.}
                                  \end{array}
                                \right. \nonumber \eeq
We will also write $\sgn(p;y_0^u:y_0^v)={\rm gap}(C(\hat{y}_i^u)^a:C(\hat{y}_i^v)^{\hat{a}})$ and refer it as the time-lag of $p$.
\end{enumerate}
\end{defn}

\section{Important Lemmas}

\begin{lem}\label{l.l.1}
Refer to Definition \ref{d.l.1}. Let $l^u$ and $l^v$ be 2 distinct closed curves and $\vec{y}$ and $\vec{\varrho}$ are parametrizations of each one respectively. Then,
\begin{align}
\lim_{\kappa \rightarrow \infty}&\frac{1}{4\pi}\frac{\kappa^3}{4}\int_{A(y)^a} ds \int_{A(\varrho)^{\hat{a}}}\  dt\ \epsilon^{ijk}\left\langle  p_\kappa^{\vec{y}_s}, p_\kappa^{\vec{\varrho}_t}\right\rangle_k y_{i,s}^\prime\varrho_{j,t}^\prime ds dt \nonumber \\
=&\pi \cdot \sum_{k=1}^3\sgn(J(C(\hat{y}_k)^a, C(\hat{\varrho}_k)^{\hat{a}}))\cdot \sgn(C(\hat{y}_k)^a:C(\hat{\varrho}_k)^{\hat{a}}) \cdot \gap(C(\hat{y}_k)^a:C(\hat{\varrho}_k)^{\hat{a}}).\nonumber
\end{align}
\end{lem}

\begin{proof}
From Item 2 in the lemma in the appendix found in \cite{EH-Lim02},
\begin{align}
\frac{\kappa^2}{4}\langle p_\kappa^{\hat{y}_{l,s}}, p_\kappa^{\hat{\varrho}_{l,t}} \rangle
&= \frac{2\pi\kappa^2}{8\pi}e^{-\kappa^2|\hat{y}_{l,s} - \hat{\varrho}_{l,t}|^2/8}. \nonumber
\end{align}

Then the integral
\beq \frac{1}{4\pi}\frac{\kappa^3}{4}\int_{A(y)^a} ds \int_{A(\varrho)^{\hat{a}}}\  dt\ \epsilon^{ijk}\left\langle  p_\kappa^{\vec{y}_{s}}, p_\kappa^{\vec{\varrho}_{t}}\right\rangle_k y_{i,s}^\prime\hat{y}_{j,t}^{\prime} ds dt \nonumber \eeq
becomes
\begin{align}
\pi\sum_{k=1}^3\int_{A(y)^a}\ ds\int_{A(\varrho)^{\hat{a}}}\  dt\ \left[y_s^\prime \times \hat{y}_t^\prime \right]_k \frac{\kappa^2}{8\pi}e^{-\frac{\kappa^2}{8}|\hat{y}_{k,s} - \hat{\varrho}_{k,t}|^2} \frac{\kappa}{\sqrt{2\pi}}\left\langle q_\kappa^{y_{k,s}}, \partial_0^{-1}q_\kappa^{\varrho_{k,t}} \right\rangle \frac{\kappa}{\sqrt{2\pi}}\left\langle \partial_0^{-1}q_\kappa^{y_{0,s}}, q_\kappa^{\varrho_{0,t}} \right\rangle. \nonumber
\end{align}

Make a change of variables: $\hat{y}_k: s \in A(y)^a  \mapsto x_+^k = \hat{y}_k(s) \in \bR^2$ and $z^k: x_+^k \mapsto  y_k(s)$, $\tau^k: x_+^k \mapsto y_{0}(s)$. Similarly, $\hat{\varrho}_k: t \in A(\varrho)^{\hat{a}}  \mapsto x_-^k =  \hat{\varrho}_k(t)\in \bR^2$ and $\tilde{z}^k: x_-^k \mapsto \varrho_k(t)$, $\tilde{\tau}^k: x_- \mapsto \varrho_0(t)$. Therefore, \beq \Xi^k: (s, t) \longmapsto y_k(s) + \varrho_k(t). \nonumber \eeq

Note that $\left[y_s \times \varrho_t\right]_k = J(\Xi^k(s,t))$. Therefore the integral becomes ($d\omega^k = dx_+^k \wedge dx_-^k$.)
\begin{align}
\pi \cdot \sum_{k=1}^3 \sgn(J(C(\hat{y}_k)^a, C(\hat{\varrho}_k)^{\hat{a}}))&\int_{ C(\hat{y}_k)^a \times C(\hat{\varrho}_k)^{\hat{a}}} \ \frac{\kappa^2}{8\pi}e^{-\kappa^2|x_+^k - x_-^k|^2/8} \nonumber \\
& \times \frac{\kappa}{\sqrt{2\pi}}\left\langle q_\kappa^{ z^k(x_+^k)}, \partial_0^{-1}q_\kappa^{ \tilde{z}^k(x_-^k)} \right\rangle\cdot \frac{\kappa}{\sqrt{2\pi}}\left\langle \partial_0^{-1}q_\kappa^{\tau^k(x_+^k)}, q_\kappa^{\tilde{\tau}^k(x_-^k)} \right\rangle\ d\omega^k. \nonumber
\end{align}
When $C(\hat{y}_k)^a \cap C(\hat{\varrho}_k)^{\hat{a}} = \emptyset$, the double integral over $C(\hat{y}_k)^a \times C(\hat{\varrho}_k)^{\hat{a}}$ goes to 0. The only case is when $p \in C(\hat{y}_k)^a \cap C(\hat{\varrho}_k)^{\hat{a}} $.

Now, we have \beq \int_{C(\hat{y}_k)^a \times C(\hat{\varrho}_k)^{\hat{a}}} \frac{\kappa^2}{8\pi}e^{-\kappa^2|x_+^k - x_-^k|^2/8} d\omega \rightarrow 1 \nonumber \eeq as $\kappa \rightarrow \infty$.

Apply Item 1 in the lemma in the appendix found in \cite{EH-Lim02},
\begin{align}
&\pi \cdot \sgn(J(C(\hat{y}_k)^a, C(\hat{\varrho}_k)^{\hat{a}}))\int_{ C(\hat{y}_k)^a \times C(\hat{\varrho}_k)^{\hat{a}}} \ \frac{\kappa^2}{8\pi}e^{-\kappa^2|x_+^k - x_-^k|^2/8} \nonumber \\
&\hspace{4.5cm}\times \frac{\kappa}{\sqrt{2\pi}}\left\langle q_\kappa^{ z^k(x_+^k)}, \partial_0^{-1}q_\kappa^{ \tilde{z}^k(x_-^k)} \right\rangle\cdot\frac{\kappa}{\sqrt{2\pi}}\left\langle \partial_0^{-1}q_\kappa^{\tau^k(x_+^k)}, q_\kappa^{\tilde{\tau}^k(x_-^k)} \right\rangle\ d\omega \nonumber \\
& \longrightarrow_{\kappa \rightarrow \infty} \pi\cdot\sgn(J(C(\hat{y}_k)^a, C(\hat{\varrho}_k)^{\hat{a}}))\sgn(C(\hat{y}_k)^a:C(\hat{\varrho}_k)^{\hat{a}})\cdot \gap(C(\hat{y}_k)^a:C(\hat{\varrho}_k)^{\hat{a}}). \nonumber
\end{align}
\end{proof}

\section{Surfaces}


\begin{notation}(Parametrization of a surface)\label{n.s.3}\\
Choose an orientable, closed and bounded surface $S \subset \bR^4$, with or without boundary. If it has a boundary $\partial S$, then $\partial S$ is assumed to be a time-like hyperlink. Do note that we allow $S$ to be disconnected, with finite number of components. Parametrize it using \beq \vec{\sigma}: (t, \bar{t}) \in I^2 \mapsto \left(\sigma_0(t,\bar{t}), \sigma_1(t,\bar{t}), \sigma_2(t,\bar{t}), \sigma_3(t,\bar{t}) \right)\in  \bR^4 \nonumber \eeq and let \beq J_{\alpha\beta} = \frac{\partial \sigma_\alpha}{\partial t} \frac{\partial \sigma_\beta}{\partial \bar{t}} - \frac{\partial \sigma_\alpha}{\partial \bar{t}} \frac{\partial \sigma_\beta}{\partial t}. \nonumber \eeq

When we project $S$ inside $\bR^3$ as $\pi_0(S)$, we can parametrize it using $\sigma$. Let
\begin{align*}
K_\sigma(\hat{t}) \equiv& K_\sigma(t,\bar{t}) := (J_{01}(t,\bar{t}), J_{02}(t,\bar{t}), J_{03}(t,\bar{t})),\\
J_\sigma(\hat{t}) \equiv& J_\sigma(t,\bar{t}) := \frac{\partial}{\partial t}\sigma(t, \bar{t}) \times  \frac{\partial}{\partial\bar{t}}\sigma(t, \bar{t}) \equiv (J_{23}(\hat{t}), J_{31}(\hat{t}), J_{12}(\hat{t})).
\end{align*}
If $S$ is assigned an orientation, we will choose a parametrization such that the inherited orientation on $\pi_0(S) \subset \bR^3$ is consistent with the vector $J_\sigma$.
\end{notation}

\begin{defn}\label{d.s.1}
Given an orientable surface $S \subset \bR^4$ as described above and a time-like loop $l \subset \bR^4$, both disjoint, project them onto $\bR^3$, denoted by $\pi_0(S)$ and $\pi_0(l)$ respectively, such that $\pi_0(l)$ intersect $\pi_0(S)$ at finitely many points. We also assume that $\partial S$ and $l$ together form a time-like hyperlink. Note that when we project $S$ inside $\bR^3$, $\pi_0(S)$ may not be a surface. But we only consider ambient isotopic equivalent classes of surface. Thus, for each point $q \in S$, we can choose a smaller surface $q \in S' \subset S$, such that $\pi_0(S')$ is indeed a surface inside $\bR^3$. So we shall assume that $\pi_0(S)$ is a surface.

Let $\pd(\pi_0; l, S)$ denote the set of finitely many intersection points between $\pi_0(l)$ and $\pi_0(S)$, henceforth termed as piercings. Let $\vec{\varrho} \equiv (\varrho_0, \varrho): I \rightarrow \bR^4$ denote a parametrization for $l$ and $\vec{\sigma}: I^2 \rightarrow \bR^4$ denote a parametrization for $S$. With such a parametrization, we orientate the curve $l$ and the surface $S$.

For each $p \in \pd(\pi_0;l,S)$, let $\sgn(p; \varrho', J_\sigma)$ denote the sign of $\varrho' \cdot J_\sigma$, whereby $J_\sigma$ was defined in Notation \ref{n.s.3}. We will also refer $\sgn(p; \varrho', J_\sigma)$ as the orientation of $p$.

Define \beq \sgn(p; \varrho_0, \sigma_0) =
\left\{
  \begin{array}{ll}
    1, & \hbox{$\varrho_0< \sigma_0$;} \\
    -1, & \hbox{$\varrho_0> \sigma_0$.}
  \end{array}
\right.
 \nonumber \eeq and refer to it as the height of $p$.

Define the algebraic piercing number of $p$ as \beq \varepsilon(p) :=  \sgn(p; \varrho', J_\sigma)\cdot \sgn(p; \varrho_0, \sigma_0). \nonumber \eeq And we define the linking number between $l$ and $S$ as \beq {\rm lk}(l, S) := \sum_{p \in \pd(\pi_0; l,S)}\varepsilon(p). \nonumber \eeq If $L = \{l^u\}_{u=1}^n$ is a time-like hyperlink, we define the linking number between $L$ and $S$ as \beq {\rm lk}(L, S) := \sum_{u=1}^n {\rm lk}(l^u, S). \nonumber \eeq


In fact, we can project $S$ inside $\bR \times \Sigma_i$ using the projection $\pi_i$. The linking number between $l$ and $S$ is independent on how we project $l$ and $S$ inside a hypersurface in $\bR^4$. Thus, we could have chosen to project them inside $\bR \times \Sigma_i$ and obtain the linking number between $l$ and $S$.

Project $l$ and $S$ inside $\bR \times \Sigma_i$ using $\pi_i$, $i=1, 2, 3$. Let $\pd(\pi_i; l, S)$ denote the set of finitely many piercings between $\pi_i(l)$ and $\pi_i(S)$. For each $p \in  \pd(\pi_i; l, S)$, we can define an algebraic piercing number $\varepsilon(p)$ in a similar fashion as above. The sum of these algebraic intersection numbers will also yield the linking number between $l$ and $S$, which is a topological invariant. 

Finally, when we are given a surface $S_0 \subset \bR^3$, we will always assume that $S_0 \subset \{0\} \times \bR^3$.
\end{defn}



\end{document}